\documentclass[12pt,nodisplayskipstretch]{article}
\usepackage[top=1.25in,bottom=1.25in,left=1.25in,right=1.25in]{geometry}
\usepackage[utf8]{inputenc} 
\usepackage{amsmath}
\usepackage{amssymb}
\usepackage{amsthm}
\usepackage{bbm}
\usepackage{xcolor}
\usepackage{graphicx}
\usepackage{float}
\usepackage{enumerate}
\usepackage{natbib}
\usepackage{setspace,lipsum}
\usepackage{bm}
\usepackage{booktabs}

\usepackage{mathrsfs}
\usepackage{comment}
\usepackage{apptools}
\usepackage{url}
\newcommand{\RN}[1]{  \textup{\uppercase\expandafter{\romannumeral#1}}}

\usepackage{booktabs}

\usepackage[subrefformat=parens]{subcaption}
\usepackage{mathtools}
\usepackage{multirow}

\theoremstyle{remark}
\theoremstyle{remark}

\theoremstyle{definition}
\newtheorem{defn}{\protect\definitionname}
\newtheorem{lem}{\protect\lemmaname}
\theoremstyle{plain}
\newtheorem{cor}{\protect\corollaryname}
\theoremstyle{plain}
\newtheorem{prop}{\protect\propositionname}
\theoremstyle{plain}

\theoremstyle{plain}
\newtheorem{thm}{\protect\theoremname}
\theoremstyle{definition}
\newtheorem{example}{\protect\examplename}

\providecommand{\examplename}{Example}
\providecommand{\axiomname}{Axiom}
\providecommand{\corollaryname}{Corollary}
\providecommand{\definitionname}{Definition}
\providecommand{\examplename}{Example}
\providecommand{\lemmaname}{Lemma}
\providecommand{\propositionname}{Proposition}
\providecommand{\theoremname}{Theorem}

\usepackage{etoolbox}
\usepackage[table]{xcolor}

\makeatletter
\patchcmd{\@maketitle}
  {\ifx\@empty\@dedicatory}
  {\ifx\@empty\@date \else {\vskip3ex \centering\footnotesize\@date\par\vskip1ex}\fi
   \ifx\@empty\@dedicatory}
  {}{}
\patchcmd{\@adminfootnotes}
  {\ifx\@empty\@date\else \@footnotetext{\@setdate}\fi}
  {}{}{}
\makeatother

\providecommand{\definitionname}{Definition}

\AtAppendix{
\numberwithin{definition}{section}
\numberwithin{theorem}{section}
\numberwithin{proposition}{section}
\numberwithin{lemma}{section}
\numberwithin{corollary}{section}}

\usepackage{enumitem}
\setlist[enumerate]{itemsep=0pt, topsep=5pt}

\title{Creative Ownership in the Age of AI}
\author{Annie Liang\footnote{Department of Economics, Northwestern} \quad \quad \quad Jay Lu\footnote{Department of Economics, UCLA}}

\begin{document}

\maketitle

\begin{abstract}
    Copyright law focuses on whether a new work is ``substantially similar'' to an existing one, but generative AI can closely imitate style without copying content, a capability now central to ongoing litigation. We argue that existing definitions of infringement are ill-suited to this setting and propose a new criterion: a generative AI output infringes on an existing work if it could not have been generated without that work in its training corpus. To operationalize this definition, we model  generative systems as closure operators mapping a corpus of existing works  to an output of new works. AI generated outputs are \emph{permissible} if they do not infringe on any existing work according to our criterion. Our results characterize structural properties of permissible generation and reveal a sharp asymptotic dichotomy: when the process of organic creations is light-tailed, dependence on individual works eventually vanishes, so that regulation imposes no limits on AI generation; with heavy-tailed creations, regulation can be persistently constraining.
\end{abstract}

\newpage

\section{Introduction}

Generative artificial intelligence systems are increasingly used to produce creative content including novels \citep{Brodsky2024aibusiness}, comics \citep{Andersen2022AltRightAI,Schor2024AndersenStability}, and even entire digital persons \citep{Heritage2025TillyNorwood}.  Trained on large body of existing works, systems can generate creations that resemble the tone or structure of particular authors, or images that evoke specific artistic traditions, without reproducing any identifiable work. For example, the following text (produced by ChatGPT-5.2) does not appear in Cormac McCarthy's \emph{The Road}, but is instantly recognizable as an imitation:
\begin{quote}
    They reached a bridge by late morning, the creek beneath it sliding past, carrying nothing that could be named. She stopped and stood there for a time, listening, though there was nothing to hear. The boy asked if they were still going the right way. She said yes. She said it because the word still mattered. They crossed and kept walking.
\end{quote}
Under standard copyright doctrine, which requires ``substantial similarity,'' such outputs typically do not constitute infringement. Yet the output may critically depend on the author's work, in the sense that it could not have been generated without their contributions to the training corpus. A wave of recent litigation---including \emph{Andersen v.\ Stability AI} and \emph{The New York Times v.\ OpenAI}---reflects growing discontent with existing notions of creative ownership.

This paper makes three primary contributions. First, we propose a new criterion for evaluating infringement on an existing work. Second, we characterize the structure of the permissible set of  generations. Finally, we characterize the evolution of permissibility as the corpus of works grows large, thus studying the restrictiveness of regulation on generative AI in saturated markets.

In our framework (described in Section \ref{sec:Framework}), a \emph{creation} is modeled as a point in $\mathbb{R}^d$, which might be a novel, painting, or even artistic character. A \emph{generator} maps a corpus of existing creations to a set of generable outputs. Formally, a generator is a closure operator: it satisfies preservation (existing creations remain generable), monotonicity (larger corpora expand the set of generable outputs), and idempotence (applying the generator to its own output produces nothing new). This abstraction encompasses various generative processes, from interpolation to feature recombination, while remaining agnostic to how the generator operates internally.

Our definition of infringement is counterfactual: a generated output \emph{infringes on} an existing creation if and only if that output could not have been generated without this creation in the training corpus. The \emph{violation set} consists of those outputs that infringe on at least one existing creation, and the \emph{permissible set} is the complement. 

Section \ref{sec:Properties} establishes structural properties of the permissible set. It expands monotonically as the corpus grows and is closed under further generation: combining permissible outputs cannot produce a violation (Proposition \ref{prop:InternalStructure}). We provide a sufficient condition for non-emptiness based on the \emph{Radon number} from convex geometry (Corollary \ref{corr:radon}). We also show that comparative statics depend on the status of new creations: incorporating a previous violation strictly expands the permissible set, while incorporating an already-permissible creation leaves it unchanged (Proposition \ref{prop:CompStatics}).

Creative corpora are not static: new works are constantly being produced and added to the training data of generative models. Section \ref{sec:LongRun} turns to the evolution of the permissible set. Our main result characterizes how the permissible set evolves as the corpus grows large. We show that when new creations are drawn from a distribution with light tails---so that extreme realizations are exponentially rare---the ratio of permissible to generable creations converges to one almost surely (Theorem \ref{thm:asym}). Intuitively, the corpus becomes sufficiently rich that multiple generative paths exist to any output, and no single creation in the corpus remains essential. By contrast, under heavy-tailed innovation, even very large corpora may continue to contain works that are genuinely distinct.  As a result, a positive measure of outputs may remain violations indefinitely. Whether infringement risk vanishes or persists thus turns on whether creative production generates redundancy or is truly innovative.

These results relate to ongoing debates about generative AI and intellectual property. Standard infringement tests ask whether an output \emph{resembles} a protected creation; our counterfactual criterion asks whether an output \emph{relies} on that creation. The two can diverge: an output may be a violation without copying its source, and an output may closely resemble a creation yet remain permissible if alternative paths to generation exist.

\subsection{Legal and Technological Context} \label{sec:Law}

\subsubsection{Existing Copyright Law and Implications for Generative AI}

Existing copyright law was not developed with generative artificial intelligence in mind, but its core rules nonetheless shape how generative systems are treated today. In most jurisdictions, copyright gives authors exclusive rights over reproduction, distribution, and the creation of derivative works (\citealp{USC106}). These rights potentially affect generative AI at two stages: model training and model output.

At the training stage, generative models are typically trained on large datasets that may include copyrighted works. Under traditional copyright principles, the ``copying'' of copyrighted works into data can constitute infringement unless it falls within an exception or limitation. In the United States, this question is largely governed by the fair use doctrine (\citealp{USC107}). Courts applying fair use consider factors such as whether the use is transformative, whether it is commercial, and whether it harms existing or potential markets for the original works (\citealp{Campbell1994, AuthorsGuildGoogle2015}). In prior cases, courts have found large-scale copying to be fair use when the copied material is not meaningfully disclosed to the public and the copying serves a distinct analytical function, such as search or indexing (\citealp{Leval1990, AuthorsGuildGoogle2015}).

While global policy responses to generative AI have largely focused on the input stage, there are arguments in favor of redirecting regulatory attention to outputs (\citealp{zhang_input_2025}). At this stage, infringement analysis follows more familiar rules. Copyright protects expression, not ideas, facts, or styles. An AI-generated output therefore infringes only if it is ``substantially similar" to protected elements of a copyrighted work (\citealp{Arnstein1946}). Most generative outputs will not meet this standard, since they are produced through statistical recombination rather than direct copying. That said, large models can sometimes memorize and reproduce portions of their training data, making output-based infringement a practical risk rather than a purely hypothetical one (\citealp{Carlini2021, Carlini2023, Nasr2023}).

Taken together, existing copyright law does not treat generative AI as inherently unlawful. Instead, legal exposure depends on two main questions: how model training is evaluated under fair use or similar doctrines, and whether generated outputs cross the line from general inspiration into protected expression.

\subsubsection{Criticisms and Ongoing Debates}

A central criticism of applying existing copyright law to generative AI concerns the distinction between style and expression. Copyright law has long drawn a firm line between protectable expression and unprotectable ideas, methods, systems, or artistic styles (17 U.S.C.\ §102(b); \citealp{Baker1879, Nichols1930}). As a result, copying an author's or artist's style---without copying protected expressive elements---does not amount to infringement, even when that style is distinctive or commercially valuable (\citealp{Satava2003, Sag2024}). This boundary reflects a longstanding policy choice to preserve creative freedom, competition, and follow-on innovation (\citealp{Leval1990}).

Generative AI systems, however, are particularly well-suited to producing outputs that closely mimic stylistic features of  human creators while avoiding literal copying. Large language models can produce text that resembles an author's tone, pacing, or narrative structure, while image models can generate works that evoke an artist's distinctive visual approach. Under current doctrine, such outputs generally fall outside copyright protection so long as they do not reproduce identifiable expressive elements from specific works (\citealp{Sag2024, Mantegna2024}).

Critics argue that this legal distinction fits poorly with the capabilities of generative AI. In many creative fields, style is often the primary source of value. A legal framework that protects expression but not style may allow AI systems to generate close substitutes that reduce demand for human creators, even when no infringement occurs in the legal sense. As a result, copyright law may permit significant displacement while offering limited remedies to affected creators (\citealp{PasqualeSun2024, USCO2025a, deRassenfosse2024}).

Defenders of the existing framework respond that extending copyright protection to style would represent a sharp break from settled law. Granting exclusive rights over stylistic features risks blurring the distinction between ideas and expression and could give creators control over broad creative territories. Such an expansion would be difficult to define and enforce, and could inhibit follow-on creativity and innovation (\citealp{Leval1990, Sag2024}). From this view, the fact that many AI-generated outputs fall outside infringement doctrine reflects copyright’s intentionally limited scope rather than a flaw in the system.

This disagreement has fueled an ongoing debate over whether generative AI exposes a genuine gap in copyright protection or simply makes longstanding limits more visible. Some scholars argue that attempting to treat stylistic imitation as infringement would push copyright beyond its intended role (\citealp{Mantegna2024}). It is also highly subjective and whether a work meets the criterion of substantial similar relies often on expert testimony.  Others maintain that generative AI alters the scale, accuracy, and substitutability of stylistic imitation in ways that weaken the original rationale for denying protection to style, even if the formal legal categories remain unchanged (\citealp{PasqualeSun2024, deRassenfosse2024}).

\subsection{Related Literature}

Besides the legal literature discussed above, our paper contributes to a growing literature in economics on copyright and intellectual property in the context of generative AI (see \citet{Lutes2025AIcopyright} and \citet{deRassenfosse2024} for excellent surveys). Our paper in particular addresses the following problem, as described in \citet{deRassenfosse2024}:
\begin{quote}
``An essential issue in the copyright treatment of creative machines is the extent to which their outputs might be legally `derivative' of the copyrighted works used to train the machine."
\end{quote}
\noindent Our paper presents a novel  formalization of this concept, where a work is legally derivative of another if it could not have been generated without the original work in its training data.

Because large-scale generative models are recent, economic theory on their implications for intellectual property remains limited. Two early contributions are \citet{Gans2024} and \citet{yang2025generativeaicopyrightdynamic}, which study optimal copyright policy from different perspectives.

\citet{Gans2024} emphasizes bargaining frictions and contracting feasibility. He demonstrates an important difference between ``small'' models trained on identifiable, contractible corpora and ``large'' models trained on web-scale data for which ex ante licensing is infeasible. In the former case, copyright protection strengthens creators’ incentives and improves welfare; in the latter, the welfare effects are ambiguous, motivating an ex post liability regime that permits training but allows harmed rights-holders to seek compensation. 

 \citet{yang2025generativeaicopyrightdynamic} adopt a dynamic perspective and study the interaction between two policy instruments: the fair-use standard governing training data and the copyrightability of AI-generated outputs. In their model, these policies jointly shape creators’ incentives to supply new data and firms’ incentives to invest in model quality. A key insight is that generous fair use can be welfare-improving when training data are abundant but counterproductive when new human-created data are scarce.

Our focus is complementary but distinct. We study the more basic question of whether the previous notions of copyright are appropriate in a setting with generative AI, and how it might alternatively be defined. We view our proposed criterion as a building block that can be embedded into broader economic models of licensing, litigation, or regulation.

Finally, Theorem \ref{thm:asym} connects to a broader literature on the nature of creative processes. Several papers have found that dividends to creative work---such as revenue from films, citations of academic papers, or attention to artworks---follow heavy-tailed distributions, with a small number of successes accounting for the bulk of observed outcomes \citep{rosen1981superstars, chung1994stochastic, devany1999uncertainty}. Like us, \citet{drugov2020tournament}, considers the process of innovation itself, and shows that whether its distribution is light- or heavy-tailed has sharp implications for the optimal design of contests and incentives. More broadly, several theories explain why creative processes may exhibit heavy tails. \citet{weitzman1998recombinant} models idea generation as a process of recombinant expansion that can generate growth exceeding exponential rates, while \citet{Lutes2025AIcopyright} argues that the use of AI as a tool in human creation may itself amplify variance in creative output.

\section{Model} \label{sec:Framework}

\subsection{Setup}

A \emph{creation} is a vector $c\in\mathbb{R}^{d}$.
A \textit{corpus} is any (Borel measurable) subset of creations $C\subset\mathbb{R}^{d}$. Examples include:

\begin{example}[Novels]
Large language models trained on corpora of published books can generate new books. A recent surge in AI-generated books led Amazon's Kindle store to impose limits on the number of books that could be self-published on the platform each week \citep{Brodsky2024aibusiness}.
\end{example}

\begin{example}[Cartoons]
Image generators such as Stable Diffusion are trained on datasets like LAION-5B, which contains billions of images scraped from the web. In \emph{Andersen v.\ Stability AI} (2023), artists allege that these models can reproduce their distinctive styles without authorization.
\end{example}

\begin{example}[Actors]
Companies such as Synthesia and HeyGen train models on video recordings of human performers to generate realistic digital avatars. The use of such likenesses—particularly of deceased or non-consenting actors—was a central issue in the 2023 SAG-AFTRA strike and subsequent negotiations over AI provisions in performer contracts \citep{SAGAFTRA2024AIWorkforce}. 
The controversy has intensified with the emergence of fully AI-generated performers such as AI talent studio Xicoia's actress Tilly Norwood \citep{Heritage2025TillyNorwood}.

\end{example}

A \emph{generator} is a technology for producing new creations from an existing corpus. Formally, it is a map from  the set of all possible corpora $\mathcal{C}$ into itself, which moreover satisfies three properties. 
\begin{defn} The map
$g:\mathcal{C}\rightarrow\mathcal{C}$ is a \textit{generator} if for all $C,D\in \mathcal{C}$ it satisfies:
\begin{enumerate}
\item Preservation: $C\subseteq g(C)$
\item Monotonicity: $C\subseteq D$ implies $g(C)\subseteq g(D)$
\item Idempotence: $g\left(g\left(C\right)\right)=g\left(C\right)$
\end{enumerate}
\end{defn}

\noindent We interpret $g(C)$ as the creations that are generable under $g$ when $g$ is trained on the works in $C$. A generator therefore (1) reproduces all input creations, (2) weakly expands the set of generable creations when the original corpus expands, and (3) is complete in that running it  on its output produces no new creations. Such maps are also known as closure or consequence operators. Below, we present several examples.

\begin{example}[Convex Hull Generator] \label{ex:CHex}
The AI firm generates creations by taking convex combinations of existing creations. We call the \emph{convex hull} generator 
\[
g_{conv}\left(C\right):=\text{conv}\left(C\right).
\]
which maps each $C$ to its convex hull.

For example, GPT-5.2 suggests the following as a convex combination of Shakespeare's \emph{Sonnet 18 (``Shall I compare thee to a summer's day?'')} and Emily Dickenson's \emph{Because I could not stop for death}:
\begin{quote}
I weigh the hours against the light, \\
They warm, then quietly recede; \\
What time bestows, it also bends, \\
Yet leaves a trace we read.
\end{quote}
Of this creation, GPT-5.2 writes: ``This poem softens Shakespeare’s rhetorical grandeur, relaxes Dickinson’s extreme compression, [and] blends temporal imagery with moderate regularity of rhythm. No single feature is taken wholesale from either poem; instead, tone, structure, and imagery lie between the two sources.''
\end{example}

\begin{example}[Splice Generator]
The AI firm generates creations by recombining features from existing
creations. We call this the \emph{splice} generator 
\[
g_{splice}\left(C\right):=\left\{ x\in\mathbb{R}^{d}\text{ : } \forall k=1,\dots, d, \mbox{ there is a } x' \in C \mbox{ s.t. } x_{k}=x'_{k}\right\} .
\]
For example, GPT-5.2 suggests the following as a splice of the poems mentioned in Example \ref{ex:CHex}:
\begin{quote}
The years—extend their measured grace—\\
As Seasons—pass me by—\\
So long as breath may utter praise—\\
I walk—yet do not die.
\end{quote}
Of this creation, GPT-5.2 writes: ``[the] diction and thematic structure (“years,” “seasons,” endurance) are drawn from Sonnet 18, [while] meter, punctuation, stanza shape, and syntactic compression are drawn from Dickinson.''
\end{example}

\begin{example}[Box Generator] The AI firm has access to both the convex hull generator and the splice generator. The \emph{box} generator composes these two maps: 
\[
g_{box}\left(D\right):=g_{conv}\left(g_{splice}\left(D\right)\right).
\]
For example, GPT-5.2 suggests the following as a box generation given the poems mentioned in Example \ref{ex:CHex}:
\begin{quote}
Time offers warmth, then slips away, \\
A season lent, not owned— \\
Yet something spoken, lightly held, \\
Outlasts the flesh and bone.
\end{quote}
Of this creation, GPT-5.2 writes: ``[this poem] inherits Shakespeare’s concern with endurance beyond time, [and] reflects Dickinson’s restraint and metaphysical inwardness, but smooths across multiple spliced variants to yield a coherent, unified style.''
\end{example}

The convex hull and 
box generators both have the property of outputting sets that are convex. They are thus  examples of what we will call convex-valued generators. 

\begin{defn}
A generator $g$ is \emph{convex-valued} if $g(C)$ is convex
for all $C\in\mathcal{C}$. 
\end{defn}

Convex-valuedness ensures that if two creations can be generated,
then so can any convex combination of them. This rules out gaps in
the space of outputs and reflects the continuity of most real generative
processes. For instance, AI systems that create digital actors or
composite images can smoothly interpolate among existing examples. We characterize convex-valued generators in Appendix \ref{app:ConvexValued}, and show in particular that composing the convex-hull generator with any other generator yields a convex-valued generator. Thus one can view the assumption of a convex-valued generator as restricting to settings where the convex hull generator
is available.

Our approach to modeling AI generation is intentionally broad and agnostic with respect to the specific algorithms employed in the process. In fact, our definition is not limited to generative AI but could equally apply to human creative processes as well. Moreover, the third property, idempotency, is to a certain degree without loss of generality. If the creative process $g$ is not idempotent, we can define a new operator $g' := g \circ g \circ \cdots$ that repeatedly iterates the original process. With diminishing returns under repeated generative iterations, $g'$ will approach idempotency. We can then treat $g'$ as the effective generator in our analysis.

\subsection{Permissible Generation} \label{sec:Infringement}

Our central definition of infringement specifies when a generated creation critically depends on the presence of some existing creation. To formalize this, we consider a counterfactual question: would the creation still be generable without that existing work in the corpus?

\begin{defn}
Fix a generator $g$ and corpus $C$. For each existing creation $c \in C$:
\begin{enumerate}
    \item The $c$-\emph{permissible set} is $p_{g}\left(c,C\right):=g\left(C\backslash \{c\}\right)$, i.e., all creations that are generable from $C \backslash \{c\}$.
    \item The $c$-\emph{violation set} is $v_{g}\left(c,C\right):=g\left(C\right)\backslash p_{g}\left(c,C\right)$, i.e., all creations that are generable from $C$ but not from $C\backslash \{c\}$.
\end{enumerate}
\end{defn}
Intuitively, the $c$-violation set consists of creations whose generability relies essentially on the presence of $c$ in the corpus. If a generated creation is a $c$-violation, then the author of $c$ has natural basis to claim infringement---the output would not exist without their contribution to the training data.

\begin{defn} Fix a generator $g$ and corpus $C$.
The \emph{permissible set} is 
\[
p_{g}\left(C\right):=\bigcap_{c\in C}p_{g}\left(c,C\right)
\]
The \emph{violation set} is 
\[
v_{g}\left(C\right):=\bigcup_{c\in C}v_{g}\left(c,C\right)
\]
\end{defn}

The permissible set consists of all creations that can be generated under $g$ from $C$, and that would still be generable were any specific existing creation to be removed from the corpus. All remaining generable creations lie in the violation set and depend essentially on some existing creation. 

Figure \ref{fig:violationCH} illustrates these sets for the convex hull generator.\footnote{For the convex hull generator, the permissible set corresponds to points with a Tukey depth of level 1.} Given existing works $c_1,c_2,\dots,c_5$, any work in the convex hull of these points is generable under $g_{conv}$. But in Panel (a), the generability of the creations in $V_{c_1}$ requires the presence of $c_1$ as an input to the generator, while the creations in $P_1$ do not. Panel (b) depicts the generable creations that are permissible, i.e., which do not depend on \emph{any} given existing work. The remaining generable creations are violations of some existing work. For the convex hull generator, the violation set is loosely an outer shell of the space of generable creations. In particular, no extreme points of $g(C)$ are permissible creations. This property is not however true for all generators; for example, Figure \ref{fig:violationBox} shows that the extreme point $c_4$ is permissible under the box generator.

In general, what the permissible set looks like depends on how the generator aggregates the underlying data and on what other existing works are present. Section \ref{sec:Properties} presents basic properties of the permissible set, and Section \ref{sec:LongRun} studies the long-run evolution of the permissible set as new creations are added to the corpus.

\begin{figure}[h]
\begin{center}
    \includegraphics[scale=0.25]{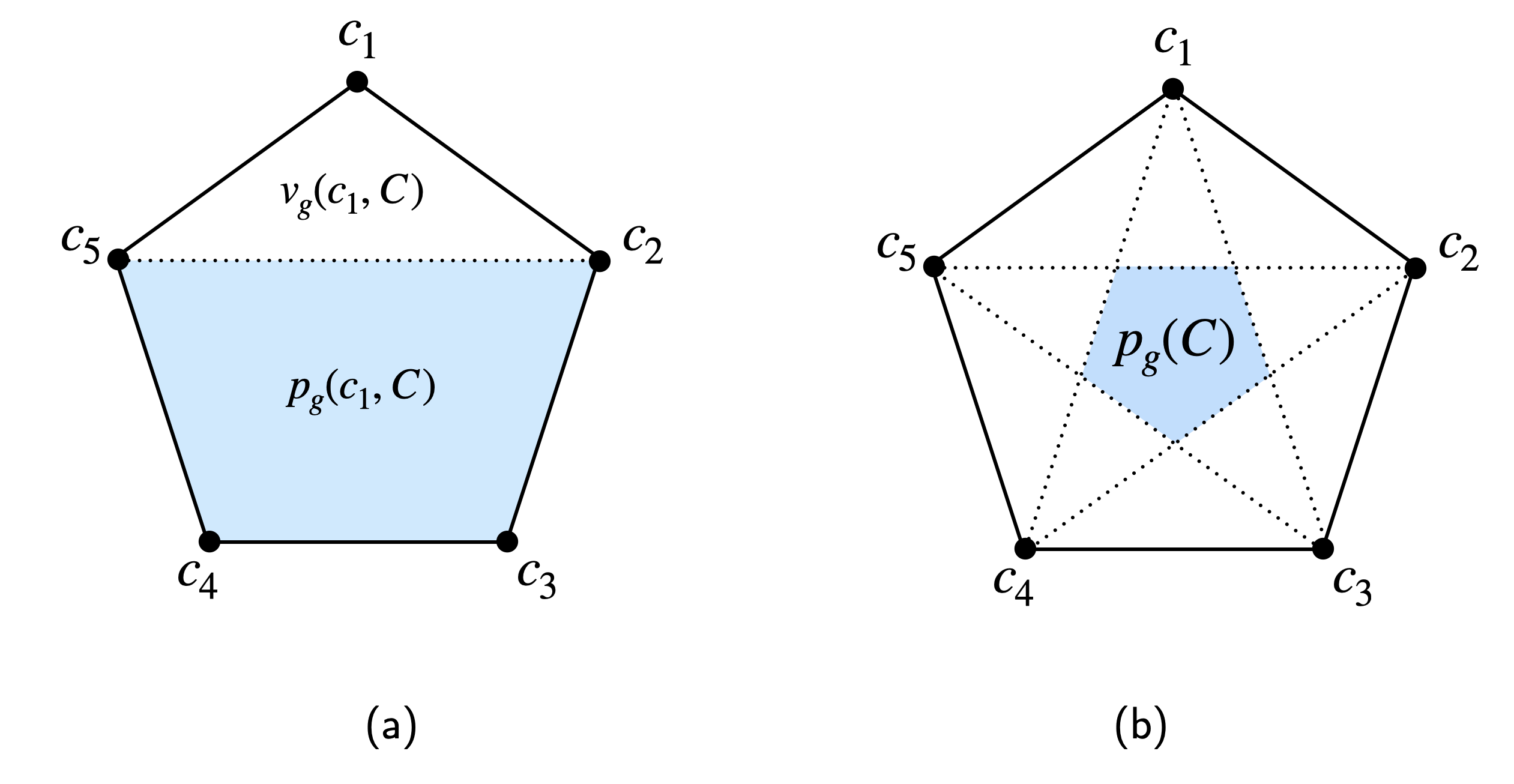}
    \caption{Consider the convex hull generator $g_{conv}$ and the corpus $C=\{c_1,c_2,\dots,c_5\}$. \emph{Panel (a):} $V_{c_1}$ is the set of $c_1$-violations (which are only constructible using $c_1$), and $P_{c_1}$ is its complement. \emph{Right:} $P = \bigcap_{c=1}^5 P_{c_i}$ is the set of points that are not $c_i$-violations for any $i=1,2,\dots,5$.} \label{fig:violationCH}
    \end{center}
\end{figure}

\begin{figure}[h]
\begin{center}
        \includegraphics[scale=0.25]{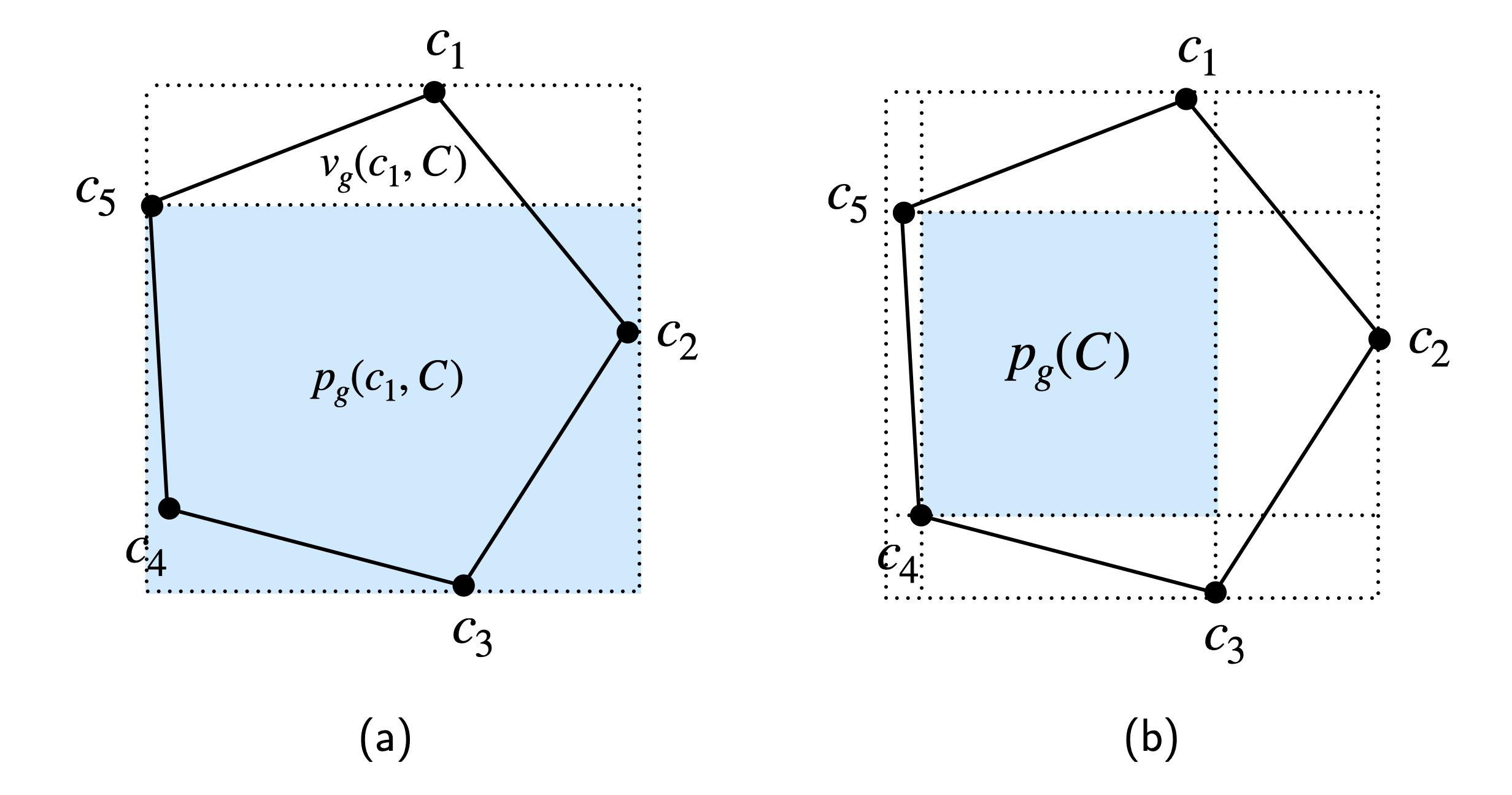}
    \caption{Replication of Figure \ref{fig:violationCH} for the box generator $g_{box}$.} \label{fig:violationBox}
    \end{center}
\end{figure}

\subsection{Discussion} \label{sec:Relationship}

As reviewed in Section \ref{sec:Law}, existing copyright doctrine applies to generative AI at two distinct stages. First, training may violate copyright law if copyrighted works are included in the training corpus. Second, generated outputs may violate copyright law if they are \emph{substantially similar} to an existing work. Our counterfactual definition departs from both of these standards. Below, we clarify how it relates to---and differs from---these prevailing legal approaches.

\paragraph{Restrictiveness of the criterion.}
Our definition can be either more or less restrictive than existing copyright standards. Relative to rules that prohibit training on copyrighted material altogether, our criterion is more restrictive in settings where AI systems are built by fine-tuning or composing existing models. For example, suppose a downstream generator can be written as $g_2 = g' \circ g_1$. Then even if the downstream firm does not itself train on copyrighted material, outputs of $g_2$ can still be violations if they depend essentially on some copyrighted work used in the training corpus of $g_1$. At the same time, our criterion does not automatically treat the presence of copyrighted works in the training corpus as problematic: inclusion is permitted so long as no generated output depends essentially on any particular copyrighted work.

Compared to the \emph{substantial similarity} standard applied to outputs, our definition again cuts in both directions. An output that is clearly not a copy of any existing work may nevertheless constitute a violation if its generation critically depends on a specific work in the corpus. Conversely, close resemblance to an existing work does not necessarily imply a violation if the output can be generated through multiple alternative paths that do not rely on that work.

Our criterion is intended as a complement rather than a substitute for existing copyright doctrines. Effective regulation may ultimately require the simultaneous enforcement of all three standards.

\paragraph{Implementation.}
Unlike ``substantial similarity,'' counterfactual generability has a clear meaning in principle: an output depends on a given work if it would no longer be producible once that work is removed from the training data. In the most direct implementation, this would involve retraining the model on \( C \setminus \{c\} \) and checking whether the output persists. While this brute-force method is usually computationally infeasible, a growing literature on \emph{machine unlearning} develops practical techniques for approximating the effect of removing individual datapoints from trained models \citep{DBLP:journals/corr/abs-1912-03817,Xiong2024MachineUnlearning}. Closely related are influence-function methods, which estimate the marginal contribution of a single training example to a model’s predictions without full retraining \citep{Koh2017Understanding}.

An alternative approach is to treat counterfactual generability as an evidentiary question and rely on qualitative proxies. For example, one could compare outputs across independently trained systems. If models trained without the plaintiff's work nonetheless produce outputs similar to a claimed violation, this suggests that the plaintiff's work was not in fact essential to generating that output. Conversely, if comparable outputs consistently fail to appear when the plaintiff’s work is absent from the training data, that absence provides evidence that the work played a necessary role.

\bigskip

Finally, although our analysis is motivated by copyright law, the criterion we propose extends beyond strictly legal settings and beyond a binary determination of infringement. More generally, dependence (as we have defined it) can be understood as a measure of whether a prior work meaningfully contributed to a given output. In academic research, for example, such dependence might inform whether a work should be cited as a predecessor. In other contexts, it could determine whether a creator has a legitimate claim to recognition or to a share of the rents generated by downstream use.

Our framework also applies to settings in which the underlying objects are not literally copyrighted. Consider the generation of digital actors or persons: while an individual’s identity is not protected by copyright, it may nonetheless be morally or socially relevant to ask whether a generated character appropriates an existing person’s voice, personality, or likeness. In this sense, the framework provides a general approach to assessing attribution and appropriation, even outside the formal boundaries of intellectual property law.

\section{Properties of Permissible Creation} \label{sec:Properties}

This section clarifies the structure of the permissible set $p_{g}(C)$. Section \ref{sec:InternalStructure} shows that permissibility behaves like a ``core'' of generative capability, in the sense that the permissible set is monotone and self-stable. Section \ref{sec:NonEmpty} provides a sufficient condition for nonemptiness of the permissible set. Section \ref{sec:ComparativeStatics} considers comparative statics in $p_{g}(C)$ as the corpus $C$ is augmented with an additional creation. All proofs can be found in the Appendix.

\subsection{Structure of the Permissible Set} \label{sec:InternalStructure}

\begin{prop} \label{prop:InternalStructure}
  For every generator $g$, the permissible set satisfies:
  \begin{enumerate}

      \item Monotonicity: If $C \subseteq D$ then $p_{g}(C) \subseteq p_{g}(D)$
      \item Stability: $g(p_{g}(C)) = p_{g}(C)$
  \end{enumerate}
\end{prop}

Part (1) studies the effect of expanding the set of existing works. Because the generator is monotone, a larger corpus expands the set of generable creations. At the same time, it introduces new works whose rights must be respected, potentially tightening the constraints on what can be produced. The result shows that the first effect dominates: expanding the corpus weakly enlarges the set of permissible creations. Intuitively, this is because if a creation can be generated in a way that does not rely on any particular existing work, then adding additional source works cannot introduce such a dependence. For example, if an AI novel generator is trained on novels from the 19th century and produces an output that does not essentially rely on any 19th century novel, then adding 20th century novels to the AI's training set will not make this novel dependent on some novel in the expanded corpus.

Part (2) says that the permissible set is closed under the generator: applying $g$ to permissible creations yields only permissible creations. This follows because if every input to the generator remains generable without $c$, then so too must the output. Closure thus captures the  natural principle that infringement cannot arise from combining non-infringing works.

\subsection{Non-Emptiness of the Permissible Set} \label{sec:NonEmpty}

We now turn to the question of existence: under what conditions is the permissible set non-empty? In general, it may be that \emph{all} generable creations are violations of some existing work, as illustrated in the following simple example:

\begin{example} Let $d=1$ so that each creation is a point on the real line. The corpus is $C =\{c_1,c_2\} = \{0,1\}$ and the generator is the convex hull generator $g:=g_{conv}$. Then $p_{g}(c_1,C) = \{1\}$ while $p_{g}(c_2,C)=\{0\}$, so the permissible set is empty: $p_{g}(C) = p_{g}(c_1,C) \cap p_{g}(c_2,C) = \emptyset$.    
\end{example}

We identify a sufficient condition for the permissible set to be nonempty, based on the \emph{Radon number} of the generator defined as follows.

\begin{defn}[Radon Number]
The \emph{Radon number} of a generator $g$, denoted $R(g)$, is the smallest integer $k$ such that any corpus $C$ of size at least $k$ contains two disjoint and nonempty subsets $A, B \subseteq C$ ($A \cap B = \emptyset$) satisfying:
\[ g(A) \cap g(B) \neq \emptyset \]
If no $k$ exists, set $R(g)=\infty$.
\end{defn}

The Radon number is a standard complexity measure in convex geometry. For the convex hull generator in $\mathbb{R}^d$, Radon's Theorem establishes that $R(g_{conv}) = d+2$, as illustrated in Figure \ref{fig:Radon}.

\begin{figure}[h] 
\begin{center}
\includegraphics[scale=0.3]{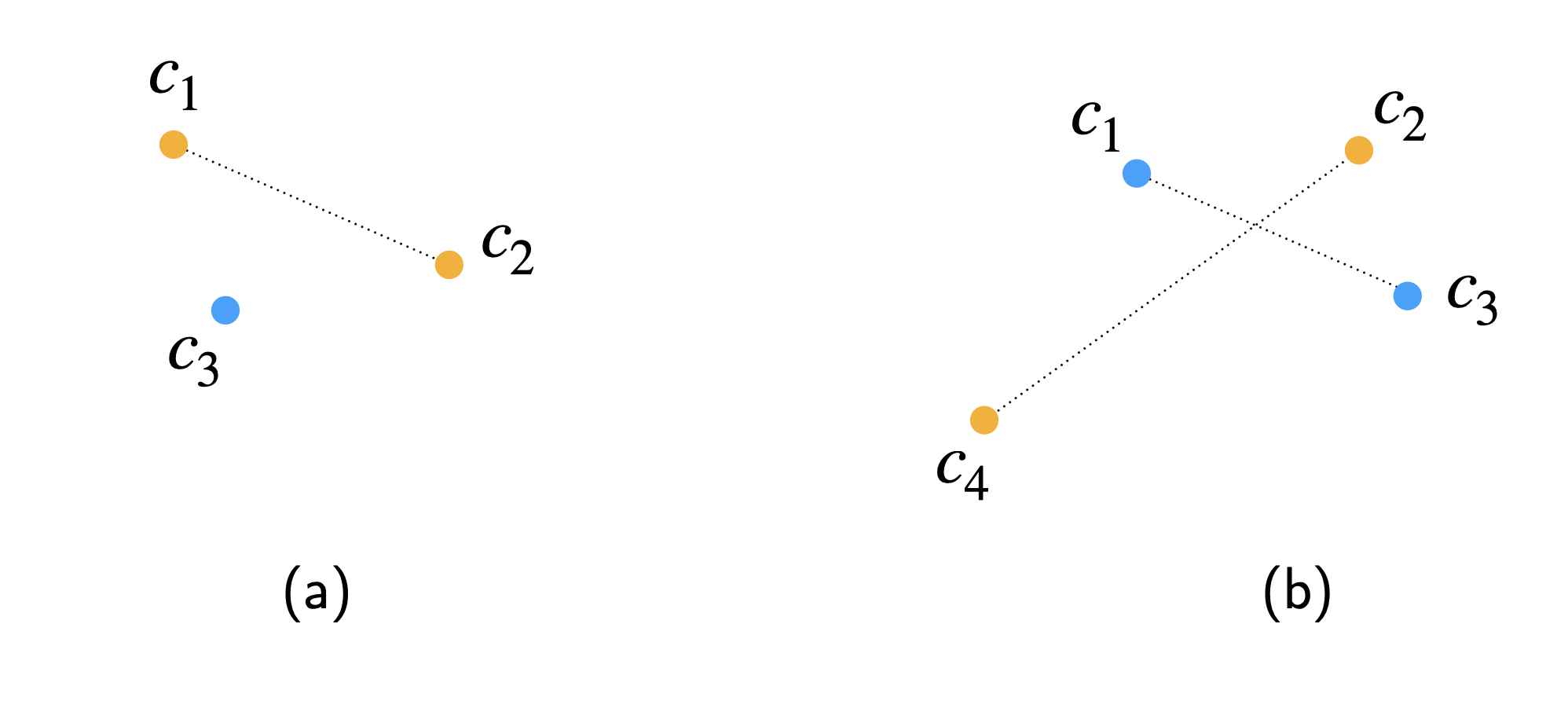}
\caption{Let $d=2$, so that creations are points in $\mathbb{R}^2$. \emph{Panel (a):} No three points in general position can be separated into disjoint subsets whose convex hulls overlap, so the Radon number is at least three; \emph{Panel (b):} Every set of four (or more) points can be separated in this way, for example in the figure let $A=\{c_1,c_3\}$ and $B=\{c_2,c_4\}$.} \label{fig:Radon}
\end{center}
\end{figure}

\begin{prop} \label{prop:radon} Fix any generator $g$. If $|C| \geq R(g)$, then the permissible set is non-empty: $p_{g}(C) \neq \emptyset$.
\end{prop}

\begin{cor} \label{corr:radon} Fix any convex-valued generator $g$. If $\vert C \vert \geq d+2$, then the permissible set is nonempty: $p_{g}(C) \neq \emptyset$.
\end{cor}

\subsection{Comparative Statics} \label{sec:ComparativeStatics}

The next result tells us how permissibility changes when new creations enter the world and are included in the set of existing works whose intellectual property must be respected. By Part (a) of Proposition \ref{prop:InternalStructure} (monotonicity of the permissible set), we know that $p_{g}(C) \subseteq p_{g}(C\cup \{c\})$ for any addition $c$, so the set of permissible creations weakly expands. The proposition below provides a more precise account depending on $c$.

\begin{prop} \label{prop:CompStatics} Fix a generator $g$  and corpus $C$. Let $c$ be a new creation not already in $C$. Then exactly one of the following cases applies:
\begin{enumerate}
\item If $c$ is permissible, the permissible set remains unchanged: 
\[c\in p_{g}(C) \quad \Longrightarrow \quad p_{g}(C) = p_{g}(C \cup \{c\})\]
\item If $c$ is a violation, the permissible set strictly expands:  
\[c \in v_{g}(C) \quad \Longrightarrow \quad  p_{g}(C) \subsetneq p_{g}(C \cup \{c\})\]
\item If $c$ is not generable, the effect on the permissible set is ambiguous: 
\[c \notin g(C) \quad \Longrightarrow \quad p_{g}(C) \subseteq p_{g}( C \cup \{c\}) \]
with equality or strict inclusion both possible. 
\end{enumerate}
    
\end{prop}

Part (1) says that including a new permissible generation in the set of accepted existing works does not expand the set of generations that are permissible. On the other hand, Part (2) says that adding a \emph{violation} to the set of existing works strictly expands the permissible set. Intuitively, when a creation that was previously a violation is added to the corpus, it becomes a primary source in its own right. This breaks the counterfactual dependence on whatever original work it had relied upon: the generator now has a second path to that part of the creative space.

Finally, Part (3) considers the introduction of a genuine innovation which was not generable from the original corpus. In this case the effect on the permissible set is ambiguous and we provide an example of each case below.

\begin{example} \label{ex:BothPossible} Let $d=2$ so that creations are points in $\mathbb{R}^2$. 

\begin{enumerate}
    \item The corpus is $C= \{c_1,c_2,c_3\}=\{(-1,0),(0,0),(1,0)\}$ and the generator is the convex hull generator. Consider the addition of the creation $c_4=(0,1)$, as depicted in Panel (a) of Figure \ref{fig:Example}. Then 
    \[p_{g}(C)= p_{g}(C \cup\{c_4\}) = \{(0,0)\}\]
    so the inclusion of $c_4$ does not expand the permissible set.
    \item The corpus is $C= \{c_1,c_2,c_3\}=\{(0,0),(0,1),(1,0)\}$ and the generator is the convex hull generator. Consider the addition of the creation $c_4=(1,1)$, as depicted in Panel (b) of Figure \ref{fig:Example}. Then
    \[p_{g}(C)= \varnothing \subsetneq \{(1/2,1/2)\} = p_{g}(C \cup \{c_4\})\]
    so the inclusion of $c_4$ strictly expands the permissible set.
\end{enumerate}

\end{example}

\begin{figure} 
\begin{center}
\includegraphics[scale=0.25]{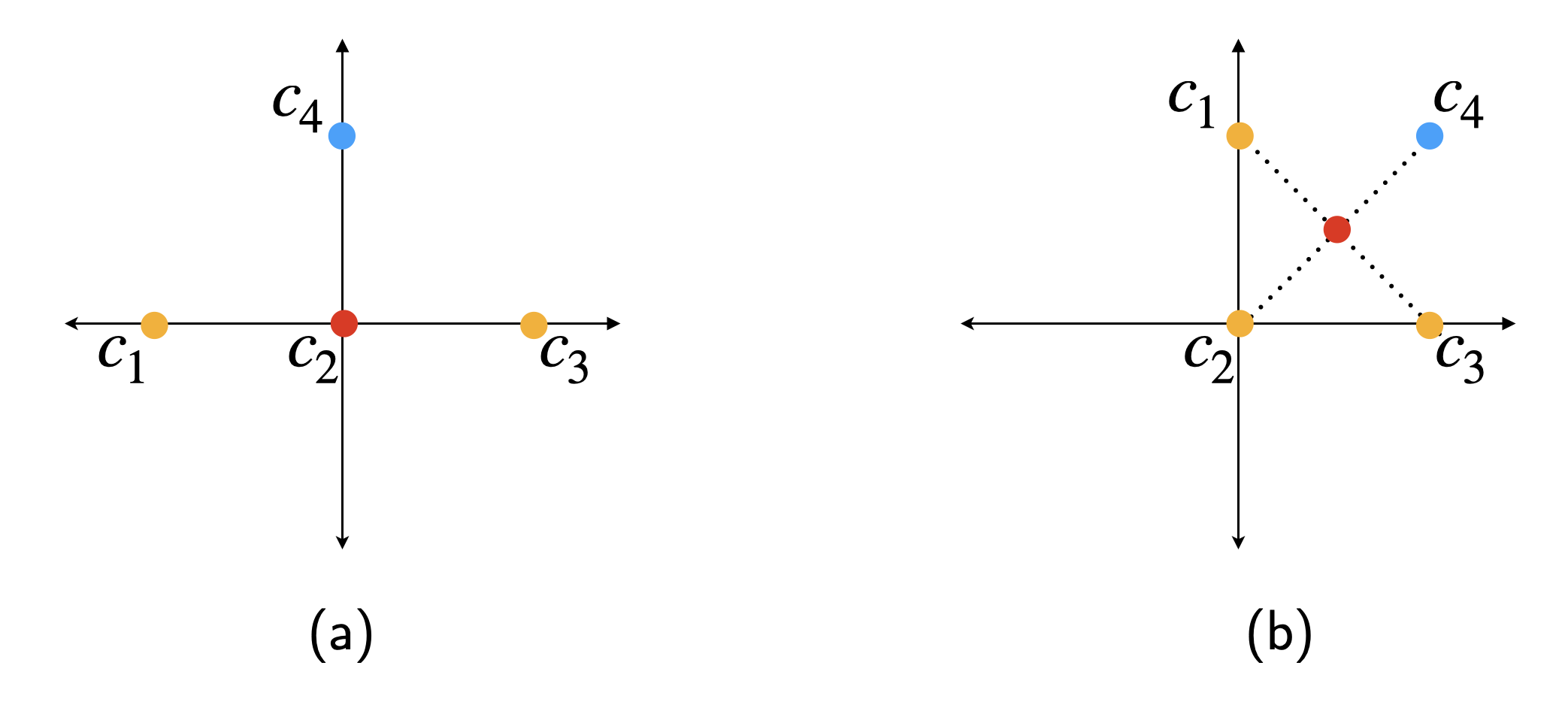}
\caption{Two examples illustrating how adding a work can affect the permissible set under the convex hull generator. In Panel (a), adding $c_4$ leaves the permissible set unchanged: $p_g(\{c_1,c_2,c_3\})= \{c_2\}=p_g(\{c_1,c_2,c_3,c_4\})$. In Panel (b), adding $c_4$ strictly expands the permissible set: $p_g(\{c_1,c_2,c_3\})=\varnothing$ while $p_g(\{c_1,c_2,c_3,c_4\})$ is the singleton intersection point of the line segments $\overline{c_1c_3}$ and $\overline{c_2c_4}$, as indicated in red.}
 \label{fig:Example}
\end{center}
\end{figure}

\section{Permissible Creations in the Long Run} \label{sec:LongRun}

The previous section characterized the permissible set for given  corpora. But creative corpora are not static: new works are constantly being produced and added to the training data of generative models. This raises a natural question: when the corpus is large, does infringement become a vanishing concern---with the permissible set eventually encompassing nearly all generable creations---or does it remain a persistent feature of generative technology?

This section shows that the answer depends on the nature of creative innovation. When new works push the frontier only ``gradually'' (in a suitably defined sense), the permissible ratio converges to one: almost everything generable becomes permissible. But when breakthrough innovations regularly occur, a positive fraction of generable creations can remain violations even as the corpus grows without bound.

To understand the extent of permissible generation, we will consider the fraction of generable creations that are permissible. Formally, let $\mathrm{Vol}(A)$ denote the Lebesgue measure (or \emph{volume}) of a set $A \subset \mathbb{R}^d$. The \emph{permissible
ratio} is defined as 
\[
r_{g}\left(C\right):=\frac{\mbox{Vol}\left(p_{g}\left(C\right)\right)}{\mbox{Vol}\left(g\left(C\right)\right)}
\]
i.e., the proportion by volume of permissible creations that lie in
the generable set.\footnote{When $g$ is convex-valued, then generically
$\mbox{Vol}\left(g\left(C\right)\right)>0$ and so this ratio is well-defined.}

We study the following creative process. Let $X_{i}$ be an i.i.d. random variable on $\mathbb{R}^{d}$ with some
strictly positive density. For each $n \in \mathbb{N}$, define the random corpus $C_{n}=\left\{ X_{1},\dots,X_{n}\right\} $.

Recall that a random variable (on $\mathbb{R}$) has \textit{light
tails} if its hazard rate is increasing.\footnote{Formally, if we let $F$ denote the cdf of the random variable and
$h\left(t\right)=\frac{f\left(t\right)}{1-F\left(t\right)}$, then
$F$ has light tails if $h\left(\cdot\right)$ is increasing.} In higher dimensions, tail behavior must account for how mass decays in different directions.  Let $U := \{ u \in \mathbb{R}^{d} : \|u\| = 1 \}$ denote the unit sphere in $\mathbb{R}^{d}$.
We express $X \in \mathbb{R}^{d}$ in terms of its polar coordinates.  The radial coordinate is
\[
r = \|X\| \in \mathbb{R}_{+},
\]
i.e., the Euclidean norm of $X$, and the angular coordinate is
\[\theta = \frac{X}{\|X\|} \in U,\]
the unit vector pointing in the direction of $X$. Every $X \neq 0$ admits the decomposition $X = r \theta$.

\begin{defn}
$X$ has \textit{uniform light tails} if
\begin{enumerate}
\item $\left\|X\right\|$ has light tails
\item there exists a $c>0$ such that for all $r>0$, the conditional density
satisfies
\[
f\left(\theta\text{ : }\left\|X\right\|\geq r\right)\geq c
\]
 
\end{enumerate}
\end{defn}
The uniform light tails condition captures  domains where expansion of the creative frontier is gradual. Formally, condition (1) ensures that extreme realizations become increasingly rare, while condition (2) rules out ``pockets'' of transient innovation that concentrate in particular directions. Together, these conditions mean that as the corpus grows, each new work is unlikely to lie far outside the convex hull of existing works.  The standard multivariate Normal distribution satisfies uniform light tails, as does any standard elliptical distribution from a light-tailed family.

We also require a technical continuity condition on the generator.
\begin{defn}
$g$ is \textit{uniform lower semicontinuous (ULS)} if for every $\varepsilon>0$,
there is a $\delta>0$ such that for all $C\in\mathcal{C}$ where
$0\in C$,
\[
\left(1-\varepsilon\right)g\left(C\right)\subseteq g\left(\left(1-\delta\right)C\right)
\]
\end{defn}
Intuitively, ULS says that contracting the corpus slightly causes only a slight contraction in the generable set. This ensures that the generative frontier does not depend too sensitively on the precise location of any single training point. Positive homogeneity---i.e., $g(\alpha C) = \alpha g(C)$ for all $\alpha > 0$---is a sufficient condition, and is satisfied by the convex hull, splice, and box generators.

We now state our main asymptotic result.
\begin{thm}
\label{thm:asym}Fix a convex-valued and ULS generator $g$. If $X_{i}$ has
uniform light-tails, then a.s. 
\[
\lim_{n\rightarrow\infty}r_{g}\left(C_{n}\right)=1
\]
\end{thm}
Theorem \ref{thm:asym} says that in creative domains with gradual innovation, copyright infringement becomes a vanishing concern as the corpus grows. The violation set---creations whose generability depends essentially on some specific existing work---shrinks to measure zero relative to the full generable set. Intuitively, when new works arrive gradually, the corpus eventually becomes rich enough that every generable creation has multiple ``generative paths'' leading to it; removing any single work from the training data leaves the creation still reachable.

This result implies that in creative markets where innovation is incremental, regulation (in the sense of our proposed criterion) only has bite in the early stages of a creative domain. If it is imposed sufficiently late, or if sufficiently many generations have already flooded the market, then very few creators will be able to defend critical dependence on their work.

The theorem also clarifies which types of creative works are most vulnerable to losing effective intellectual property protection within our framework. Works located in the interior of the creative space---surrounded by many similar works---quickly become dispensable as training data accumulate. By contrast, works at the frontier, in sparsely populated regions of the space, may remain essential for generating nearby outputs even as the corpus grows large.

Our result extends beyond creative processes that satisfy uniform light tails. For example,  generic multivariate random variables can be normalized to identity covariance via a linear transformation (whitening). Since Theorem \ref{thm:asym} is preserved under linear transformations, any elliptical distribution from a light-tailed family (e.g., Normal, exponential) satisfies the premise. What the theorem rules out are heavy-tailed distributions, where breakthrough innovations regularly occur. The following example demonstrates that in such domains, a positive fraction of generable creations remain violations even asymptotically.
\begin{example}[Heavy tails]
Suppose $X_{i}$ are Pareto distributed on $\mathbb{R}_{+}$ and let $g=g_{\text{conv}}$ be the convex hull generator. The Pareto distribution is the canonical model of heavy-tailed phenomena, arising in contexts where ``superstar'' outcomes regularly occur---bestselling novels, viral content, or paradigm-shifting artistic innovations. Note that
\[
r_{g}\left(C_{n}\right)=\text{\ensuremath{\frac{\mbox{Vol}\left(p_{g}\left(C_{n}\right)\right)}{\mbox{Vol}\left(g\left(C_{n}\right)\right)}}}\leq\frac{\max\left\{ X_{1},\dots,X_{n-1}\right\} }{\max\left\{ X_{1},\dots,X_{n}\right\} }
\]
For Pareto distributions, the ratio of successive maxima does not converge to one almost surely. The $n$-th observation has a non-negligible probability of dramatically exceeding all previous observations, no matter how large $n$ becomes. Consequently, $r_g(C_n)$ remains bounded below one: the most recent breakthrough always defines a substantial portion of the violation set.
\end{example}

This contrast between light-tailed and heavy-tailed domains implies that, in creative fields characterized by occasional breakthroughs---such as avant-garde art and literary fiction---protection for individual works may retain its importance even as corpora grow large. Creators of paradigm-shifting works can continue to have legitimate infringement claims because their contributions remain essential for generating outputs near the creative frontier. By contrast, in mature creative markets where innovation is largely incremental---possibly including genre fiction and commercial movies---the marginal contribution of any single work diminishes rapidly as the corpus expands, making individual infringement claims correspondingly weaker.

\section{Extension: General Violations} \label{sec:groupwise}

The analysis so far has focused on a setting in which the entire training corpus consists of protected works and infringement is defined at the level of individual works. In practice, however, not all works are protected---some belong to the public domain---and copyright disputes often arise at the level of collections rather than single works. A novelist may sue over their entire catalog; a photographer
may claim infringement of their portfolio; a class action may aggregate
works by many creators. This subsection considers a more general framework that allows for unprotected works 
and groupwise violations.

A \textit{collection} of protected works is a family of subsets $\mathcal{A}\subseteq\mathcal{C}$.
In the main case we have considered thus far, $\mathcal{A}$
is the collection of all singleton works. We allow the collection of protected works to be different from the collection of all singleton works and extend the notions of
permissibility and violation correspondingly.
\begin{defn}
Fix a generator $g$, a collection of protected works $\mathcal{A}$, and a corpus $C$.
The permissible set is 
\[
p_{g}^{\mathcal{A}}\left(C\right):=\bigcap_{A\in\mathcal{A}}g\left(C\setminus A\right)
\]
and the violation set is 
\[
v_{g}^{\mathcal{A}}\left(C\right):=g\left(C\right)\setminus p_{g}^{\mathcal{A}}\left(C\right).
\]
\end{defn}
For each set of protected works $A\in\mathcal{A}$ in the collection, $g\left(C\setminus A\right)$
consists of all outputs that can be generated without using any input
from $A$, i.e. permissible with respect to $A$. The permissible
set $p_{g}\left(\mathcal{A},C\right)$ thus consists of all works
permissible with respect to every set in the collection $\mathcal{A}$. 

Proposition \ref{prop:InternalStructure} extends fully.
\begin{prop}
\label{prop:groupwise} For every generator $g$ and collection $\mathcal{A}$,
the permissible set satisfies
\begin{enumerate}
\item If $C\subseteq D$, then $p_{g}^{\mathcal{A}}\left(C\right)\subseteq p_{g}^{\mathcal{A}}\left(D\right)$ 
\item $g\left(p_{g}^{\mathcal{A}}\left(C\right)\right)=p_{g}^{\mathcal{A}}\left(C\right)$ 
\end{enumerate}
\end{prop}
Part (1) is monotonicity in the corpus: expanding the corpus weakly
expands the permissible set, just as for individual works. Part (2)
is stability: generating from permissible works yields only permissible
works.

We now consider what happens to the permissible set under different collections. We say a collection or protected works is \textit{richer} than another if for every set of protected works in the former, there is a larger set of protected works in the latter. 

\begin{defn}
A collection $\mathcal{A}$ is \textit{richer} than another collection $\mathcal{B}$ if for every $B\in\mathcal{B}$,
there exists some $A\in\mathcal{A}$ such that $B\subseteq A$. 
\end{defn}

There are two notions of what it means for a collection to have ``more" protected works. First, the collection could simply consist of a greater number of protected works. Second, the collection could consist of larger groups of protected works, e.g., novelists coming together to sue as a coalition. The first notion corresponds to greater volume while the second to greater aggregation. Our definition of richness captures both.  

The following result shows that the permissible set is smaller in richer collections.

\begin{prop} \label{prop:rich} Fix a generator
$g$ and corpus $C$. If a collection $\mathcal{A}$ is richer than a collection $\mathcal{B}$ then
\[
p_{g}^{\mathcal{A}}\left(C\right)\subseteq p_{g}^{\mathcal{B}}\left(C\right).
\]
\end{prop}

To illustrate the implications of this result, considers what happens when two human creators combine their individual creations to form a coalition. An immediate corollary is that violations are \emph{superadditive}: the
violation of a union can exceed the union of their individual violations.
\begin{cor}
[Superadditivity] Fix a generator $g$
and corpus $C$. For any two protected sets $A,B\in\mathcal{C}$,
\[
v_{g}^{\left\{ A,B\right\} }\left(C\right)\subseteq v_{g}^{\left\{ A\cup B\right\} }\left(C\right)
\]
with strict inclusion possible. 
\end{cor}

Superadditivity implies that when multiple creators join a class action, their combined violation set can exceed the sum of their individual claims. This strengthens the collective bargaining position of creator coalitions and suggests that group licensing may command a premium over the sum of individual licenses.

Another implication of Proposition \ref{prop:rich} is that removing protected works from the collection naturally expands the permissible set. Combining this with Theorem \ref{thm:asym}, we have the following asymptotic result.
\begin{cor} 
Fix a convex-valued and ULS $g$, and suppose the collection $\mathcal{A}$ consists of some subset of all singleton works. If $X_i$ has uniform light tails, then a.s.
\[\lim_{n \to \infty} \frac{\mathrm{Vol}(p_{g}^{\mathcal{A}}\left(C_n\right))}{\mathrm{Vol}(g(C_n))} = 1 \]
\end{cor}

In other words, Theorem \ref{thm:asym} is conservative: its implication that AI generation becomes unrestricted in the long run under light-tailed innovation assumes that every work in the corpus is protected. If there are works in the corpus that are unprotected (e.g., public domain), then this renders AI generation more permissible and strengthens the asymptotic result.

A natural question is to ask what happens if we permit group violations. In this case, we conjecture that as long as the number of individuals in each group is finitely bounded, then Theorem \ref{thm:asym} would still hold. This is because as $n$ grows large, any finite number of individuals banding together as a group would still be negligible and the same asymptotic forces will prevail. On the other hand, if the group size grows unboundedly with the corpus---for instance, if it includes all works by authors who have ever contributed---the limiting behavior depends on how it grows relative to $n$. Characterizing this relationship is an open question.

\section{Conclusion and Open Questions}

Our paper characterizes the structure of permissible generation but takes as given both the generator and the distribution of creative works. A natural extension is to endogenize these objects. If creators anticipate how generative systems will use existing works, they may adjust what they produce. Analyzing equilibrium in this environment requires modeling the strategic interaction between creators, who choose where in the creative space to locate new works, and AI firms, who choose which generative capabilities to develop. One plausible conjecture is that creators concentrate effort near the frontier of the creative space rather than in its interior. Such strategic positioning would thicken the tails of the creation distribution and could sustain a larger violation set than would arise under non-strategic creation.

A second open question concerns social optimality---specifically, what level of permissiveness would we ideally sustain? At one extreme, if the permissible ratio converges to one then nearly all generative outputs are non-infringing, and this may weaken creators' ex ante incentives to produce new works. At the other extreme, a large and persistent violation set preserves creator rents but limits the social gains from generative technologies. Our framework and proposed definition provide a foundation for formalizing and evaluating these tradeoffs.

\newpage

\appendix

\section*{Appendix}

\section{Preliminary Results} \label{app:ConvexValued}

\begin{prop} \label{prop:convex_v} The following are equivalent:
\begin{enumerate}
   \item $g$ is convex-valued
   \item $g=g_{conv}\circ g$
   \item $g=g\circ g_{conv}$
\end{enumerate}
\end{prop}

\begin{proof}
We will first show that $g$ is convex-valued if and only if $g=g_{conv}\circ g$.
In one direction, when $g$ is convex-valued then $g\left(C\right)$
is convex, implying $\text{conv}\left(g(C)\right)=g(C)$
as desired. In the other direction, if $g(C)=\mbox{conv}(g(C))$ then $g(C)$ is clearly convex.

We now show that $g$ is convex-valued if and only if $g=g\circ g_{conv}$. First
suppose $g$ is convex valued. Then since $C \subseteq g_{conv}(C)$ we have 
\begin{equation} \label{eq:gCsmaller}
g(C) \subseteq g(g_{conv}(C))
\end{equation} by monotonicity of $g$. Moreover,
\begin{align*}
    g_{conv}(C) & \subseteq g_{conv}(g(C)) && \mbox{since $C \subseteq g(C)$ by containment} \\
    & = g(C) && \mbox{since $g = g_{conv}\circ g$}
\end{align*}
Applying the mapping $g$ to both sides of this set inclusion yields
\begin{equation} \label{eq:gCbigger}
g(g_{conv}(C)) \subseteq g(g(C)) = g(C)
\end{equation}
where the last equality follows from idempotency. Combining (\ref{eq:gCsmaller}) and (\ref{eq:gCbigger}) yields 
$g\left(g_{conv}\left(C\right)\right)=g\left(C\right)$ as desired.

For the
converse, suppose $g=g\circ g_{conv}$. Then
\begin{align*}
    g_{conv}\left(g\left(C\right)\right) & \subseteq g( g_{conv}\left(g\left(C\right)\right)) && \mbox{by containment} \\
    & =g\left(g\left(C\right)\right) && \mbox{by assumption that $g=g \circ g_{conv}$} \\
    & = g(C) && \mbox{by idempotence} \\
    & \subseteq g_{conv}(g(C)) && \mbox{by containment} 
\end{align*}
so $g=g_{conv}\circ g$, which directly implies that $g$ is convex-valued.
\end{proof}

\begin{cor} \label{corr:CHminimal}
The convex hull generator is minimal among all convex-valued generators, i.e.,
\[
g_{conv}\left(C\right)\subseteq g\left(C\right)
\]
for all convex-valued generators $g$.
\end{cor}

\begin{proof} For any convex-valued generator $g$, 
\[g_{conv}(C) \subseteq g(g_{conv}(C)) = g(C)\]
where the set inclusion follows by containment and the inequality follows from Proposition \ref{prop:convex_v}.
\end{proof}

\section{Proofs of Results in Section \ref{sec:Properties}}

\subsection{Proof of Proposition \ref{prop:InternalStructure}}

\emph{Part (a):} Let $C\subseteq C'$ and without loss write $C' = C \cup A$ where $C$ and $A$ are disjoint sets. Then
\begin{align*}
    p_{g}(C \cup A) & = \bigcap_{c \in C\cup A} p_{g}(c,C \cup A) \\
    & = \bigcap_{c \in C\cup A} g((C\cup A)\backslash \{c\}) \\
    & = \bigcap_{c \in C} g((C\cup A)\backslash \{c\}) \cap \bigcap_{c' \in A} g((C\cup A)\backslash \{c'\})
\end{align*}
while 
\[
p_{g}\left(C\right)=\bigcap_{c\in C}g\left(C\backslash \{c\}\right)
\]
Fix an arbitrary $c\in C$. Then
$C \backslash \{c\} \subseteq \left(C \cup A\right)\backslash \{c\}$, so 
\[g\left(C \backslash \{c\}\right) \subseteq g\left(\left(C \cup A\right)\backslash \{c\}\right)\]
by monotonicity of $g$. Moreover, $C\backslash \{c\} \subseteq C \subseteq \left(C\cup A\right)\backslash \{c'\}$ for every $c' \in A$, so also
\[g\left(C\backslash \{c\}\right)\subseteq g\left(\left(C\cup A\right)\backslash \{c'\}\right).\]
Thus $p_{g}\left(C\right) \subseteq p_{g}(C \cup A) = p_{g}(C')$ as desired. \\

\emph{Part (b):} In one direction, $p_{g}(C) \subseteq g(p_{g}(C))$ by containment. In the other direction, define $E_c := g(C \backslash \{c\})$ for each $c \in C$, so that $p_{g}(C) = \bigcap_{c \in C} E_c$. Fix $c' \in C$. Since
$\bigcap_{c \in C} E_c \subseteq E_{c'},$  monotonicity of $g$ implies
\[g\left(\bigcap_{c \in C} E_c\right) \subseteq g\left(E_{c'}\right)\]
for each $c' \in C$. So
\begin{equation} \label{eq:gInclusion} g\left(\bigcap_{c \in C} E_c\right) \subseteq \bigcap_{c'\in C} g\left(E_{c'}\right)
\end{equation}
But for each $c'\in C$,
\[g(E_{c'}) = g(g(C \backslash \{c'\}))=g(C \backslash \{c'\})\]
where the second equality follows from idempotence. So
\[\bigcap_{c'\in C} g\left(E_{c'}\right) = \bigcap_{c' \in C} g(C \backslash \{c'\}) = p_{g}(C).\]
Combining this with (\ref{eq:gInclusion}) yields 
\[g(p_{g}(C)) = g\left(\bigcap_{c \in C} E_c\right) \subseteq p_{g}(C)\]
completing the proof.

\section{Proof of Proposition \ref{prop:radon}}

Since $|C| \geq R(g)$, by definition there exist disjoint subsets $A, B \subseteq C$ such that $g(A) \cap g(B) \neq \emptyset$. Let $x$ be an element in this intersection. We show that $x \in p_{g}(C)$.

Recall that $p_{g}(C) = \bigcap_{c \in C} g(C \setminus \{c\})$. Fix any arbitrary existing creation $c \in C$. Since $A$ and $B$ are disjoint, $c$ cannot belong to both sets. This leads to two cases:
\begin{enumerate}
    \item If $c \notin A$, then $A \subseteq C \setminus \{c\}$. By monotonicity, $g(A) \subseteq g(C \setminus \{c\})$. Since $x \in g(A)$, it follows that $x \in g(C \setminus \{c\})$.
    \item If $c \notin B$, then $B \subseteq C \setminus \{c\}$. By monotonicity, $g(B) \subseteq g(C \setminus \{c\})$. Since $x \in g(B)$, it follows that $x \in g(C \setminus \{c\})$.
\end{enumerate}
In either case, $x$ remains generable after the removal of $c$. Since this holds for all $c \in C$, we conclude that $x \in p_{g}(C)$.

\section{Proof of Corollary \ref{corr:radon}}

Assume $|C| \geq d+2$. By Radon's theorem, there exist disjoint subsets $A, B \subseteq C$ whose convex hulls intersect, i.e., $g_{conv}(A) \cap g_{conv}(B) \neq \emptyset$. By Proposition \ref{prop:radon}, the permissible set for the convex hull generator is non-empty, $p_{g_{conv}}(C) \neq \emptyset$.

Recall from Corollary \ref{corr:CHminimal} that $g_{conv}$ is minimal among convex-valued generators, meaning $g_{conv}(C) \subseteq g(C)$ for any corpus $C$. This inclusion extends to the permissible sets:
\[
p_{g_{conv}}(C) = \bigcap_{c \in C} g_{conv}(C \setminus \{c\}) \subseteq \bigcap_{c \in C} g(C \setminus \{c\}) = p_{g}(C).
\]
Since $p_{g_{conv}}(C)$ is non-empty, it follows that $p_{g}(C)$ is non-empty as well.

\section{Proof of Proposition \ref{prop:CompStatics}}

\emph{Part (a):} Choose some $c'\in C$. Since by assumption $c$ is permissible, 
\begin{equation} \label{eq:cIng}
c \in p_{g}\left(C\right)\subseteq g\left(C\backslash c'\right)
\end{equation}
i.e., $c$ is generable from $C$ without $c'$. Moreover,
\begin{align*}
\left(C\cup \{c\}\right)\backslash c' & =\left(C\backslash c'\right)\cup \{c\} \\
& \subseteq \left(C\backslash c'\right)\cup g\left(C\backslash c'\right) && \mbox{by (\ref{eq:cIng})}\\
 & \subseteq  g\left(C\backslash c'\right)\cup g\left(C\backslash c'\right) && \mbox{by containment}\\
  & =g\left(C\backslash c'\right)
\end{align*}

Applying $g$ to both sides yields
\begin{equation} \label{eq:gSimplify}
g\left(\left(C\cup \{c\}\right)\backslash c'\right)  \subseteq g\left(g\left(C\backslash c'\right)\right)=g\left(C\backslash c'\right)
\end{equation}
where the set inclusion follows from monotonicity of $g$, and the final equality follows from idempotence.

Thus, 
\begin{align*}
p_{g}\left(C\cup \{c\}\right) & = g(C) \cap\bigcap_{c'\in C}g\left(\left(C\cup \{c\}\right)\backslash c'\right)\\
 & \subseteq g(C)  \cap\bigcap_{c'\in C}g\left(C\backslash c'\right) && \mbox{by (\ref{eq:gSimplify})}\\
 & = \bigcap_{c'\in C}g\left(C\backslash c'\right) = p_{g}(C)
\end{align*}
as desired.

\emph{Part (b):} Let $c\in v_{g}\left(C\right)$ so $c\not\in p_{g}\left(C\right)$.
Since $g$ is monotone, $p_{g}\left(C\right)\subseteq p_{g}\left(C\cup \{c\}\right)$ by Part (a) of Proposition \ref{prop:InternalStructure}. We will show that $c\in p_{g}\left(C\cup \{c\}\right)$, which then implies that $p_{g}\left(C\cup \{c\}\right)$ is strictly larger.

Suppose otherwise. Then there must be some $c'\in C$ such that 
\begin{equation} \label{eq:cInset}
c\in v_{g}\left(c',C\cup \{c\}\right)=g\left(C\cup \{c\}\right) -  g\left(\left(C\cup \{c\}\right)\backslash \{c'\}\right)
\end{equation}
First suppose $c\neq c'$. Then 
\[
c\in\left(C\cup \{c\}\right)\backslash \{c'\}\subseteq g\left(\left(C\cup \{c\}\right)\backslash \{c'\}\right)
\]
where the set inclusion follows from containment. But this contradicts (\ref{eq:cInset}).

Suppose instead $c = c'$. Then 
\[
c\in v_{g}\left(c,C\cup \{c\}\right)=g\left(C\cup \{c\}\right) -g(C)
\]
so $c\not\in g\left(C\right)$ yielding a contradiction with our initial assumption that $c \in v_{g}(C)$.

\emph{Part (c):} The set inclusion $p_{g}(C) \subseteq p_{g}(C \cup \{c\})$ is a direct consequence of Part (a) of Proposition \ref{prop:InternalStructure}, and the possibility of both equality and strict inclusion is shown in Example \ref{ex:BothPossible}.

\section{Proof of Theorem \ref{thm:asym}}

The proof proceeds in the following steps. We first establish two key intermediate results. Lemma \ref{lem:subset} shows that almost surely, for all sufficiently large $n$, each  new point $X_n$ lies inside a slight radial expansion of the convex hull of the previous points $X_1, \dots, X_{n-1}$.  Lemma \ref{lem:all_i} shows that if removing the \emph{last} prior observation shrinks the convex hull by at most a multiplicative factor $1+\varepsilon$, then the same is true for removing \emph{any} prior observation.

Throughout this proof, let $U$ denote the unit sphere in $\mathbb{R}^{d}$ and let $\lambda$ denote the surface area measure on $U$. Recall that $C_n = \{X_1, \dots, X_n\}$ is the size-$n$ corpus, and define $K_n=\mbox{conv}(C_n)$ to be its convex hull.

\begin{lem}
\label{lem:subset} For any $\varepsilon>0$, almost surely there exists $N<\infty$ such that for all $n\ge N$,
\[
X_n \in (1+\varepsilon)K_{n-1}.
\]

\end{lem}

\begin{lem}
\label{lem:all_i} Suppose that $(Z_i)_{i \geq 1}$ is a sequence of positive random variables such that:
\begin{enumerate}
    \item For any $\varepsilon>0$, almost surely there exists $N(\varepsilon)$ such that $Z_n \le (1+\varepsilon)\max_{i\le n-1} Z_i$ for all $n \ge N(\varepsilon)$.
    \item $\max_{i \le n} Z_i \to \infty$ almost surely as $n \to \infty$.
\end{enumerate}
Then almost surely there exists $\widetilde{N}(\varepsilon)$ such that for all $n \ge \widetilde{N}(\varepsilon)$ and all $j\in\{1,\dots,n\}$,
\[
\max_{i\le n} Z_i \le (1+\varepsilon)\max_{i\in\{1,\dots,n\}\setminus\{j\}} Z_i.
\]
\end{lem}
\bigskip

We are now ready to complete the proof.  By Lemma \ref{lem:subset},
 for any $\frac{\delta}{1-\delta}>0$ there is an a.s.
event such that for all sufficiently large $n$
\begin{alignat*}{1}
X_{n} & \in\left(1+\frac{\delta}{1-\delta}\right)K_{n-1}
\end{alignat*}
Equivalently, for all $u\in U$
\begin{alignat*}{1}
u\cdot X_{n} & \leq\left(1+\frac{\delta}{1-\delta}\right)\max_{i\leq n-1}u\cdot X_{i}
\end{alignat*}

Let $\sigma\left(\cdot,u\right)$ denote the support function in direction $u$. Then almost surely for all sufficiently large $n$, all $u\in U$, and all
$j\in\{1,\dots,n\}$,
\begin{align*}
\sigma\!\left((1-\delta)\operatorname{conv}(C_n),u\right)
&= (1-\delta)\,\sigma(C_n,u)
&& \mbox{by homogeneity of $\sigma$} \\
&= (1-\delta)\max_{i\le n} u\cdot X_i
&&  \\
&\le \max_{i\in\{1,\dots,n\}\setminus\{j\}} u\cdot X_i
&& \mbox{by Lemma~\ref{lem:all_i}} \\
&= \sigma(C_n\setminus\{X_j\},u)
&&  \\
&= \sigma\!\left(\operatorname{conv}(C_n\setminus\{X_j\}),u\right)
&& \mbox{since $\sigma(\operatorname{conv}(\cdot),u)=\sigma(\cdot,u)$.}
\end{align*}

\noindent  Equivalently,
$\left(1-\delta\right)\text{conv}\left(C_n\right)\subseteq\text{conv}\left(C_n\backslash X_{j}\right)$, which implies 
\begin{alignat*}{1}
g\left(\left(1-\delta\right)\text{conv}\left(C_n\right)\right) & \subseteq g\left(\text{conv}\left(C_n\backslash X_{j}\right)\right)
\end{alignat*}
by monotonicity of $g$, and finally that
\[g\left(\left(1-\delta\right)C_n\right) \subseteq g\left(\left(C_n\backslash X_{j}\right)\right)\]
by convex-valuedness of $g$ (see Proposition \ref{prop:convex_v}). Intersecting over $j$ yields
\[
g\!\left((1-\delta)C_n\right)
\subseteq
\bigcap_{j=1}^n g\!\left(C_n\setminus\{X_j\}\right)
=
p_{g}(C_n).
\]
Moreover, $p_{g}(C_n)\subseteq g(C_n)$ by definition, and 
therefore,
\[g\!\left((1-\delta)C_n\right)\subseteq p_{g}(C_n)\subseteq g(C_n).\]

This implies that
\[\text{vol}\left(g\left(\left(1-\delta\right)C_n\right)\right)  \leq\text{vol}\left(p_{g}\left(C_n\right)\right)\leq\text{vol}\left(g\left(C_n\right)\right)\]
Dividing through by $\text{vol}(g(C_n))$ (and noting $\text{vol}(g(C_n))>0$ for large $n$), we also have
\[\frac{\text{vol}\left(g\left(\left(1-\delta\right)C_n\right)\right)}{\text{vol}\left(g\left(C_n\right)\right)}  \leq\frac{\text{vol}\left(p_{g}\left(C_n\right)\right)}{\text{vol}\left(g\left(C_n\right)\right)}\leq1\]

Consider some $\varepsilon>0$. By assumption, $g$ is ULS and moreover  $0\in\text{conv}\left(C_n\right)$ for large $n$ without loss. 
We can then find some $\delta$ such that
\[
\left(1-\varepsilon\right)g\left(\text{conv}\left(C_n\right)\right)  \subseteq g\left(\left(1-\delta\right)\text{conv}\left(C_n\right)\right)\]
\noindent Since $g$ is convex-valued, we also have
\[\left(1-\varepsilon\right)g\left(C_n\right)  \subseteq g\left(\left(1-\delta\right)C_n\right)
\]
by Proposition~\ref{prop:convex_v}. This implies that 
\[
\left(1-\varepsilon\right)^d\text{vol}\left(g\left(C_n\right)\right)=\text{vol}\left(\left(1-\varepsilon\right)g\left(C_n\right)\right)\leq\text{vol}\left(g\left(\left(1-\delta\right)C_n\right)\right)
\]
so
\[
(1-\varepsilon)^d\leq\frac{\text{vol}\left(g\left(\left(1-\delta\right)C_n\right)\right)}{\text{vol}\left(g\left(C_n\right)\right)}\leq\frac{\text{vol}\left(p_{g}\left(C_n\right)\right)}{\text{vol}\left(g\left(C_n\right)\right)}
\]
Taking $\varepsilon\rightarrow0$ yields the result.

\subsection{Proof of Lemma \ref{lem:subset}}

Since $K_{n-1}=\mbox{Conv}(X_1,\dots,X_{n-1})$ is (by definition) a convex set, it suffices to control the support function in every direction. That is, we will prove that for every $u\in U$,
\[u\cdot X_n \;\le\; (1+\varepsilon)\max_{i\le n-1} u\cdot X_i.\]
We separate this argument into two steps. First, we control the magnitude of the new draw $X_n$, showing that its norm is eventually bounded by a radius $r_n$ determined by the tail behavior of $\| X \|$. Second, we control the geometry of the existing convex hull $K_{n-1}$ at this same scale, showing that the hull already extends a  distance comparable to $r_n$ in every direction.  Combining these two facts implies that no new point of magnitude comparable to $r_n$ can exceed the previous maximum in any direction, yielding $X_n\in(1+\varepsilon)K_{n-1}$ for all sufficiently large $n$.

\subsubsection{Controlling the magnitude of the new draw $X_n$} Fix an arbitrary constant $K>0$ (to be chosen later). Let $F$ be the cdf of $\left|X\right|$ and $\bar{F}:=1-F$. For all $n$ such that $K\frac{\log n}{n}\in(0,1]$, define $r_n\in\mathbb{R}_+$ by
\[
\bar F(r_n)=K\frac{\log n}{n}.
\]
Such an $r_n$ exists by continuity of $\bar F$.

\begin{lem} \label{lem:MagnitudeBound} Almost surely for large $n$,
\[
\left|X_{n}\right|\leq\left(1+\delta\right)r_{n}
\]
\end{lem}

\begin{proof}
Fix $\delta>0$ and write $\bar F(t)=\mathbb P\{|X_n|>t\}$. By the light-tail assumption,
$\phi(t):=\log \bar F(t)$ is concave on $[0,\infty)$.\footnote{This is equivalent to a nondecreasing hazard rate
$h(t)=f(t)/\bar F(t)$, since $\phi'(t)=-h(t)$.}
Since $\phi$ is concave and $\phi(0)=\log\bar F(0)=0$, we have
$\phi((1+\delta)t)\le (1+\delta)\phi(t)$ for all $t\ge 0$. Therefore,
\begin{align*}
\mathbb P\{|X_n|>(1+\delta)r_n\}
&=\bar F((1+\delta)r_n)
= \exp\!\big(\phi((1+\delta)r_n)\big) \\
&\le \exp\!\big((1+\delta)\phi(r_n)\big)
= \left(\bar F(r_n)\right)^{1+\delta}
\le K^{1+\delta}\left(\frac{\log n}{n}\right)^{1+\delta}.
\end{align*}

\noindent Thus,
\[
\sum_{n}\mathbb{P}\left\{ \left|X_{n}\right|>\left(1+\delta\right)r_{n}\right\} \leq K^{1+\delta}\sum_{n}\left(\frac{\log n}{n}\right)^{1+\delta}<\infty
\]
as $\delta>0$. Since this converges, we have the
desired result by the first Borel-Cantelli lemma. 
\end{proof}

\subsubsection{Directional coverage at the frontier}

Let $V\subseteq U$ be a finite $\theta$-cover of the sphere, i.e., a finite set of
directions such that every $u\in U$ lies within angle $\theta$ of some $v\in V$:
\[
u\cdot v \ge \cos(\theta)\qquad\forall u\in U.
\]
Fix $v\in U$. We say that an observation $X_i$ is \emph{$v$-good at level $n$} if
\[
|X_i|\ge r_n
\quad\text{and}\quad
\frac{X_i}{|X_i|}\cdot v \ge \cos(\theta).
\]
Thus, a $v$-good observation is both large in magnitude and aligned with the
direction $v$ up to angle $\theta$.

By the uniform light-tails assumption, there exists a constant $c>0$ such that
for all $r>0$ and all $u\in U$,
\[
\mathbb P\!\left(\frac{X}{\| X \|}\cdot u \ge \cos(\theta)\,\middle|\,\| X \|\ge r\right)
\ge c\, s_u .
\]
In particular, this bound holds for all $v\in V$. Moreover, the surface area of
the spherical cap of directions within angle $\theta$ of $v$ is strictly
positive:
\[
s_v := \lambda\!\left\{u\in U : u\cdot v \ge \cos(\theta)\right\} > 0
\qquad \forall v\in V.
\]
Since $V$ is finite, we may choose a constant $K>0$ such that
\[
c s_v \ge \frac{4}{K} \qquad \forall v\in V.
\]

For each $v\in V$, let $A_{n-1}(v)$ denote the event that there exists at least one $v$-good observation among $\{X_1,\dots,X_{n-1}\}$. By independence of $X_i$,
\[
\mathbb P\!\left(A_{n-1}(v)^c\right) = (1-p_n)^{\,n-1},
\]
where
\[
p_n := \mathbb P\!\left(\| X \|\ge r_n,\ \frac{X}{\| X \|}\cdot v \ge \cos(\theta)\right),
\qquad X:=X_1.
\]

Writing $f$ for the marginal density of $\| X \|$ and $f_r$ for the conditional
density of $X/\| X \|$ given $\| X \|=r$, we have
\begin{align*}
p_n
&= \int_{r_n}^\infty
\left(\int_{u\cdot v\ge \cos(\theta)} f_r(u)\,du\right) f(r)\,dr \\
&\ge c s_v\,\mathbb P\{\| X \|\ge r_n\}
\;\ge\; c s_v K \frac{\log n}{n}
\;\ge\; 4\,\frac{\log n}{n}.
\end{align*}
It follows that
\[
\mathbb P\!\left(A_{n-1}(v)^c\right)
\le \exp\!\left(-(n-1)p_n\right)
\le \exp\!\left(-4\frac{n-1}{n}\log n\right).
\]
Since $(n-1)/n\ge 1/2$ for large $n$,
\[
\mathbb P\!\left(A_{n-1}(v)^c\right) \le n^{-2}.
\]
Thus,
\[
\sum_n \mathbb P\!\left(A_{n-1}(v)^c\right) < \infty,
\]
and the first Borel--Cantelli lemma implies that, almost surely, there exists a
finite (random) index $N_v$ such that $A_{n-1}(v)$ holds for all $n\ge N_v$.

Since $V$ is finite, we may intersect over all $v\in V$. Hence, with probability
one, there exists
\[
N := \max_{v\in V} N_v < \infty
\]
such that for all $n\ge N$ and all $v\in V$, there exists a $v$-good observation
among $\{X_1,\dots,X_{n-1}\}$.

\subsubsection{Uniform control of the support function}

We now combine directional coverage with the magnitude bound to obtain a
uniform estimate over all directions $u\in U$. We work on the intersection of
the almost-sure events established above.

Fix $u\in U$. Choose $v\in V$ such that $u\cdot v \ge \cos(\theta)$. On the
almost-sure event above, there exists a $v$-good observation $X_i$ with
$i\le n-1$. Since both $u$ and $X_i/|X_i|$ lie within angle $\theta$ of $v$, the
angle between $u$ and $X_i/|X_i|$ is at most $2\theta$, and therefore
\[
u\cdot \frac{X_i}{|X_i|} \ge \cos(2\theta).
\]
Since $|X_i|\ge r_n$ for a $v$-good observation, it follows that
\begin{equation}\label{eq:rbound}
u\cdot X_i \ge r_n \cos(2\theta).
\end{equation}

Now fix $\varepsilon>0$. Choose $\delta>0$ and $\theta\in(0,\pi/4)$ sufficiently
small so that
\begin{equation}\label{eq:parameters}
\frac{1+\delta}{\cos(2\theta)} < 1+\varepsilon.
\end{equation}
Then almost surely for all sufficiently large $n$,
\begin{align*}
u\cdot X_n
&\le |X_n| && \text{by Cauchy--Schwarz} \\
&\le (1+\delta) r_n && \text{by Lemma~\ref{lem:MagnitudeBound}} \\
&\le \frac{1+\delta}{\cos(2\theta)}\, u\cdot X_i && \text{by \eqref{eq:rbound}} \\
&\le (1+\varepsilon)\, u\cdot X_i && \text{by \eqref{eq:parameters}} \\
&\le (1+\varepsilon)\max_{j\le n-1} u\cdot X_j.
\end{align*}
Since this bound holds uniformly for all $u\in U$, we conclude that
\[X_n \in (1+\varepsilon)K_{n-1},\]
as desired.

\subsection{Proof of Lemma \ref{lem:all_i}}

By assumption (1), since $Z_{i}>0$ a.s., we have a.s.
\[
\lim\sup_{n}\frac{Z_{n}}{\max_{i\leq n-1}Z_{i}}\leq1+\varepsilon
\]
We will show that in this event,
\[
\lim\sup_{n}\left(\max_{j\in\left\{ 1,\dots,n\right\} }\frac{\max_{i\in\left\{ 1,\dots,n\right\} }Z_{i}}{\max_{i\in\left\{ 1,\dots,n\right\} \backslash j}Z_{i}}\right)\leq1+\varepsilon
\]
Suppose otherwise, so we can find an infinite sequence $k$ and $j_{k}\leq k$
such that
\begin{alignat*}{1}
\max_{i\in\left\{ 1,\dots,k\right\} }Z_{i} & >\left(1+\varepsilon\right)\max_{i\in\left\{ 1,\dots,k\right\} \backslash j_{k}}Z_{i}\\
Z_{j_{k}} & >\left(1+\varepsilon\right)\max_{i\in\left\{ 1,\dots,k\right\} \backslash j_{k}}Z_{i}\geq\left(1+\varepsilon\right)\max_{i\leq j_{k}-1}Z_{i}
\end{alignat*}
This implies that $j_{k}\leq N$. But this implies for all $k>N$,
\[
\max_{i\leq N}Z_{i}\geq Z_{j_{k}}>\left(1+\varepsilon\right)\max_{i\in\left\{ N,\dots,k\right\} }Z_{i}
\]
But by assumption (2), $\max_{i\in\{N,\dots,k\}} Z_i \to \infty$ as $k\rightarrow\infty$, yielding
a contradiction. 

\section{Proofs for Section \ref{sec:groupwise}}
\subsection{Proof of Proposition \ref{prop:groupwise}}

\noindent\emph{Part (1).}
Suppose $C \subseteq D$. For any $i\in I$,
\[
C \setminus S_i \subseteq D \setminus S_i.
\]
By monotonicity of the generator $g$,
\[
g(C \setminus S_i) \subseteq g(D \setminus S_i)
\quad \text{for all } i\in I.
\]
Taking intersections over $i$ yields
\[\bigcap_{i\in I} g(C \setminus S_i)
\subseteq
\bigcap_{i\in I} g(D \setminus S_i),\]
which is exactly $p_{g}(\mathcal{S}, C) \subseteq p_{g}(\mathcal{S}, D)$.

\medskip
\noindent\emph{Part (2).}
Let $P := p_{g}(\mathcal{S}, C)$. By definition, $P \subseteq g(P).$ To show the reverse inclusion, fix any $i\in I$. Since
\[P \subseteq g(C \setminus S_i),\]
we have
\[g(P) \subseteq g(g(C \setminus S_i)= g(C \setminus S_i)\]
where the inclusion follows from monotonicity of $g$, and the equality follows from idempotence of $g$. Intersecting over $i$ yields
\[
g(P) \subseteq \bigcap_{i\in I} g(C \setminus S_i) = P
\]
as desired.

\subsection{Proof of Proposition \ref{prop:rich}}

Since $\mathcal{A}$ is richer than $\mathcal{B}$, let $A_{B}\in\mathcal{A}$
denote the set in $\mathcal{A}$ such that $A_{B}\supseteq B$ for every $ B \in \mathcal{B}$. By
the monotonicity of $g$, we have
\[
g\left(C\backslash A_{B}\right)\subseteq g\left(C\backslash B\right)
\]
Thus,
\[
p_{g}^{\mathcal{A}}\left(C\right)=\bigcap_{A\in\mathcal{A}}g\left(C\setminus A\right)\subseteq\bigcap_{B\in\mathcal{B}}g\left(D\setminus A_{B}\right)\subseteq\bigcap_{B\in\mathcal{B}}g\left(C\setminus B\right)=p_{g}^{\mathcal{B}}\left(C\right)
\]
as desired.

\end{document}